\documentclass{article}

\usepackage{amsmath}
\usepackage{amssymb}
\usepackage{amsthm}
\usepackage{bm}
\usepackage{array}
\usepackage{makecell}
\usepackage{dcolumn}
\usepackage{graphicx}
\usepackage{siunitx}
\usepackage{color}
\usepackage[affil-it]{authblk}

\renewcommand{\i}{\ensuremath{\text{I}}}
\newcommand{\ii}{\ensuremath{\text{I\!I}}}
\newcommand{\ici}{\ensuremath{\text{I,I}}}
\newcommand{\icii}{\ensuremath{\text{I,I\!I}}}
\newcommand{\iici}{\ensuremath{\text{I\!I,I}}}
\newcommand{\iicii}{\ensuremath{\text{I\!I,I\!I}}}

\newcommand{\iciii}{\ensuremath{\text{I,}\Gamma}}
\newcommand{\iiciii}{\ensuremath{\text{I\!I,}\Gamma}}
\newcommand{\overbar}[1]{\mkern 1.5mu\overline{\mkern-1.5mu#1\mkern-1.5mu}\mkern 1.5mu}
\newtheorem{theorem}{Theorem}
\newtheorem{lemma}[theorem]{Lemma}
\newtheorem{assumption}{Assumption}

\title{Stable absorbing boundary conditions for molecular dynamics in general domains}
\author{Xiaojie Wu \thanks{\texttt{xxw139@psu.edu}; Corresponding author}}
\author{Xiantao Li \thanks{\texttt{xxl12@psu.edu};}}
\affil{Department of Mathematics, \\Penn State University, University Park}

\begin{document}  
\maketitle
\begin{abstract}
    A new type of absorbing boundary conditions for molecular dynamics simulations are presented. The exact boundary conditions for crystalline solids with harmonic approximation are expressed as a dynamic Dirichlet-to-Neumann (DtN) map. It connects the displacement of the atoms at the boundary to the traction on these atoms. The DtN map is valid for a domain with general geometry. To avoid evaluating the time convolution of the dynamic DtN map, we approximate the associated kernel function by rational functions in the Laplace domain. The parameters in the approximations are determined by interpolations.  The explicit forms of the zeroth, first, and second order approximations will be presented. The stability of the molecular dynamics model, supplemented with these absorbing boundary conditions is established. Two numerical simulations are performed to demonstrate the effectiveness of the methods.
  \end{abstract}
  
  \section{Introduction}
Absorbing boundary conditions (ABCs) are extremely important numerical tools \cite{Neuhasuer1989, Clayton1977, Mur1981} to efficiently simulate wave propagation phenomena in large or infinite domains. In these methods, the computational cost is greatly reduced by truncating the entire domain into a much smaller region of interest. The boundaries of the truncated region are specifically treated to retain the characteristics of the full system.  Mathematically, an exact boundary condition (BC) can often be derived, which tends to be nonlocal. For instance, in the work of Engquist and Majda \cite{Engquist1977}, the exact BC is represented via pseudo-differential operators. They can be approximated by Pad\'e approximations, which lead to local ABCs.  Another popular approach is to introduce an artificial absorbing layer to facilitate the propagation to the exterior. The most well-known method of this type is the perfectly matched layer (PML) method, which is originally formulated by Berenger \cite{Berenger1994} for Maxwell's equations. 


In the context of molecular dynamics (MD), the role of the ABCs is again critical: typical MD simulations involve a huge number of atoms and the computation tends to be prohibitively expensive. Therefore, ABCs are much needed to minimize the computational cost.  It is worthwhile to first point out that  the ABCs for MD models are different from those for continuous PDEs, since the dispersion relations are quite different. 
There has been a lot of recent development of  ABCs for MD models. There are mainly three approaches:

{\it Exact ABCs}. The exact ABC can be derived for planar boundaries using Fourier/Laplace transform, or the lattice Green's functions \cite{Cai2000,Wagner2003, Wagner2004, Park2005, Karpov2005}.  The main computational challenge is due to the convolutional integral in time, which has to be evaluated at each step. To the best of our knowledge, the best solution is given by the rational approximation of the continuum Green's function \cite{Namilae2007}. Nevertheless, these exact BCs are limited to planar boundaries, and the explicit forms break down at corners. 

{\it Approximate ABCs by minimizing the total phonon reflection}. In this approach, the reflection coefficient \cite{E2002, E2001} is calculated at boundaries of the truncated region. An approximate ABC, typically with a small number of previous time steps involved, is sought by minimizing the total reflection, in the form of  energy fluxes  \cite{Li2006}. This idea is extended to the cases of finite temperature \cite{Li2007}. A similar idea is the matching boundary conditions,  proposed by Tang and his co-workers in a series of works \cite{Wang2010, Fang2012, Pang2017}. They assume that the time history kernel can be approximated by an artificial local boundary condition. The unknowns are determined by a matching procedure at some pre-selected wave numbers. 

{\it Discrete PML}.  The  PML  method was extended to  MD models by To and Li \cite{To2005, Li2006a}. Another extension based on continuum PML was performed by Guddati and co-workers \cite{Guddati2009}. Their approach is based on Guddati's previous work, perfectly matched discrete layers (PMDL) \cite{Guddati2006}, which is a more accurate implementation of PML. Similar to continuous PML, discrete PML is adaptable to complex geometries, for example, one does not need to treat the corner issue explicitly. The main subtlety is how to determine the parameters in PML. For examples, if they are too small or too large, significant reflections will occur.

In this paper, we propose a systematic approach to formulate and approximate the ABCs for domains with general geometry. In particular, we adopt an impulse/response perspective. More specifically,  
the impulse corresponds to the displacement (or traction) of the atoms at the boundary, which will induce a mechanical field in the surround region. This influence will  in turn exert a kick-back force on the atoms at the boundary. Therefore,  the response, which would exhibit a history-dependence, corresponds to these forces (or displacement). This can be formulated more precisely using the dynamic Dirichlet-to-Neumann (DtN) map. This idea has been pursued for the wave equation and Schr\"{o}dinger equation in  \cite{Jiang2008, Alpert2002, Jiang2004}.  Our proposed method involves the following steps: 
\begin{enumerate}
\item [(a)] To convert the dynamics problem to a static one, we take the Laplace transform in time. 
\item [(b)] We reduce the computational domain by using an atomistic-based boundary element method (ABEM). This eliminates the degrees of freedom associated with the surrounding atoms. 
\item [(c)] The dynamic DtN map can be obtained by an inverse Laplace transform. Instead of implementing this exact ABC, we approximate the Laplace transform by rational functions.
\item [(d)] The rational function approximation reduces the nonlocal ABCs to local ODEs, which can be easily implemented.
\end{enumerate}
We emphasize that the use of the ABEM method in (a) allows us to treat domains with general geometry, including multi-connected domains. Meanwhile, since the full MD model is a Hamiltonian system, a naive approximation of the BCs can lead to a unstable model. Our observation is that the formulation via the DtN map in step (b) makes the stability 
analysis more amenable.  The rational approximation in step (d) eliminates the need to perform an inverse Laplace transform numerically. 


The layout of this paper is as follows. The exact ABC is presented in section \ref{sec:dtn}. We  discuss the evaluation of the DtN map, along with its approximation  in section \ref{sec:approx0}. The stability of the ABCs is established in section \ref{sec:stability}. The approximate ABC is extended to the nonlinear MD model under a partial-harmonic approximation in section \ref{sec:harmonic}. We present two numerical experiments to demonstrate the effectiveness of these methods in section \ref{sec:num}. 

\section{The formulation of Absorbing Boundary Conditions}
\label{sec:dtn}
Consider a system with $N$ atoms in a domain $\Omega$. Atoms in the domain have positions denoted by $\bm X$ and displacement $\bm u$. The entire domain $\Omega$ is divided into two regions $\Omega_{\i}$ and $\Omega_{\ii}$, $\Omega_{\i}\cup \Omega_{\ii} = \Omega$, as illustrated by Fig. \ref{fig:1}. Here $\Omega_\i$ can be multi-connected regions, e.g., around multiple local lattice defects.
\begin{figure}[!htp]
  \centering
  \includegraphics[scale=0.2]{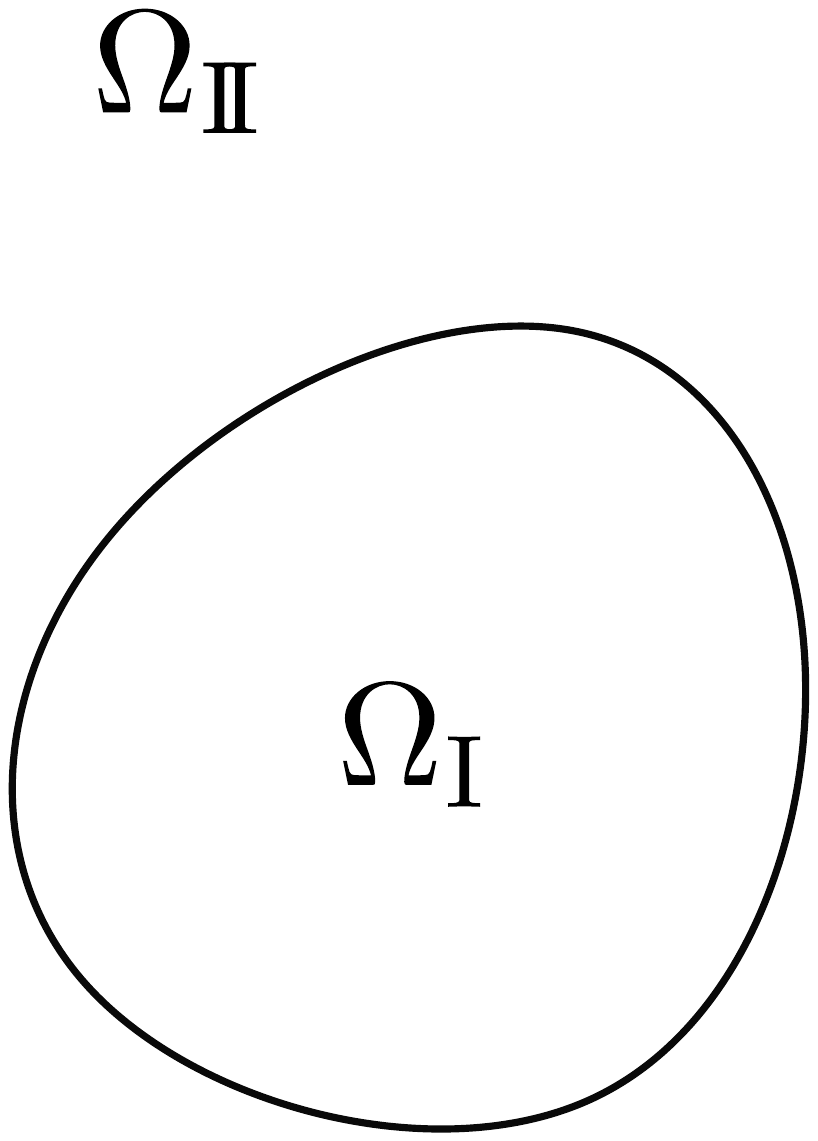}
  \caption{An illustration of the domain decomposition. }
  \label{fig:1}
\end{figure}
$\Omega_{\i}$ refers to the computational domain, where the MD model is actually being implemented, and $\Omega_{\ii}$ indicates the surrounding region that is to be removed.  We denote $\bm u_\i$ as the displacement of atoms in $\Omega_\i$ and $\bm u_\ii$ as the displacement of atoms in $\Omega_\ii$. $\dim(\bm u_\i)=n_\i$, and $\dim(\bm u_\ii)=n_\ii$. In realistic problems, $\Omega_{\ii}$ involves much more atoms than those in the region $\Omega_{\i}$., i.e., $n_\ii\gg n_\i.$

The ABCs can be derived when the interactions with $u_\ii$ are linear. The resulting BC would be effective, as long as such approximations are acceptable. For the clarity of our formulation, we first assume that {\it all} the interactions are linear (In  section \ref{sec:harmonic}, we will demonstrate that the formulation can be easily extended to the case where the interactions in $\Omega_\i$ is nonlinear). 
This is the case where the formulation of the ABCs is most transparent. In this case, the potential energy of the system is quadratic, that is,
\begin{equation}
  \label{eq:1}
  V = \frac{1}{2}\bm u^TD\bm u,
\end{equation}
where $D_{ij}=\frac{\partial^2 V}{\partial \bm u_i\partial \bm u_j}$ corresponds to the Hessian matrix  of  the exact potential energy $V$ in the MD model.  Typically these models have  short-range interactions, with the  cut-off radius denoted by $r_{\text{cut}}$.  This implies that $D = [D_{ij}]$ is a sparse matrix. In general, $D$ is symmetric positive definite, due to the stability requirement.


Atoms in region $\Omega_{\i}$ and $\Omega_{\ii}$ follow the equations of motion,
\begin{equation}
  \label{eq:lds}
  \left\{
    \begin{aligned}
      \ddot {\bm u}_{\i} &= -D_{\ici}\bm u_{\i} - D_{\icii}\bm u_{\ii},\\
      \ddot {\bm u}_{\ii} &= -D_{\iici}\bm u_{\i} - D_{\iicii}\bm u_{\ii}.
    \end{aligned}
  \right.
\end{equation}
Here, we set the mass to unity and defined $D_\ici=[D_{ij}]$, where $\bm X_i\in \Omega_\i$, and $\bm X_j\in \Omega_\i$. $D_\icii$, $D_{\iici}$, and $D_{\iicii}$ are defined similarly. We have $D_\icii = D_\iici^T$.
The dynamics \eqref{eq:lds} is stable, since the coefficient matrix
\begin{equation*}
  \begin{bmatrix}
    -D_{\ici}&-D_{\icii}\\
    -D_{\iici}&-D_{\iicii}
  \end{bmatrix}
\end{equation*}
is symmetric negative definite.

To build up the Dirichlet-to-Neumann (DtN) map, we take the Laplace transform of the second equation of \eqref{eq:lds} and get,
\begin{equation}
  \label{eq:2}
  s^2\bm U_{\ii}=-D_{\iici}\bm U_{\i}-D_{\iicii}\bm U_{\ii}, \quad s > 0,
\end{equation}
where $\bm U_{\i}(s) = \mathcal L\{\bm u_{\i}(t)\}$, and $\bm U_{\ii}(s)=\mathcal L\{\bm u_{\ii}(t)\}$, which are the Laplace transforms of $\bm u_{\i}(t)$ and $\bm u_{\ii}(t)$, respectively. Here, we assume that $\bm u_\i$ and $\bm u_\ii$ are bounded functions, such that the Laplace transformations are well-defined. In the Laplace domain, we are able to formally solve the equations because $(s^2I+D_{\iicii})$ is positive definite. The solution of Eq. \eqref{eq:2} is the mapping from $\bm U_{\ii}$ to $\bm U_\i$, that is,
\begin{equation}
  \label{eq:3}
  \bm U_{\ii}=-(s^2I+D_{\iicii})^{-1}D_{\iici}\bm U_\i.
\end{equation}
In the time domain, Eq. \eqref{eq:3} becomes,
\begin{equation}
  \label{eq:40}
  \bm u_\ii=\int_0^t\beta(t-s)\bm u_\i(s) ds,
\end{equation}
where $\beta(s) = \mathcal {L}^{-1}\{B(t)\}$. Clearly, the matrix $\beta$ has $n_\ii$ rows, and it is too large to work with. Fortunately, the interactions have short range and this representation can be simplified to only involve atoms close to the boundary.
\begin{figure}[!htp]
  \centering
  \includegraphics[scale=0.15]{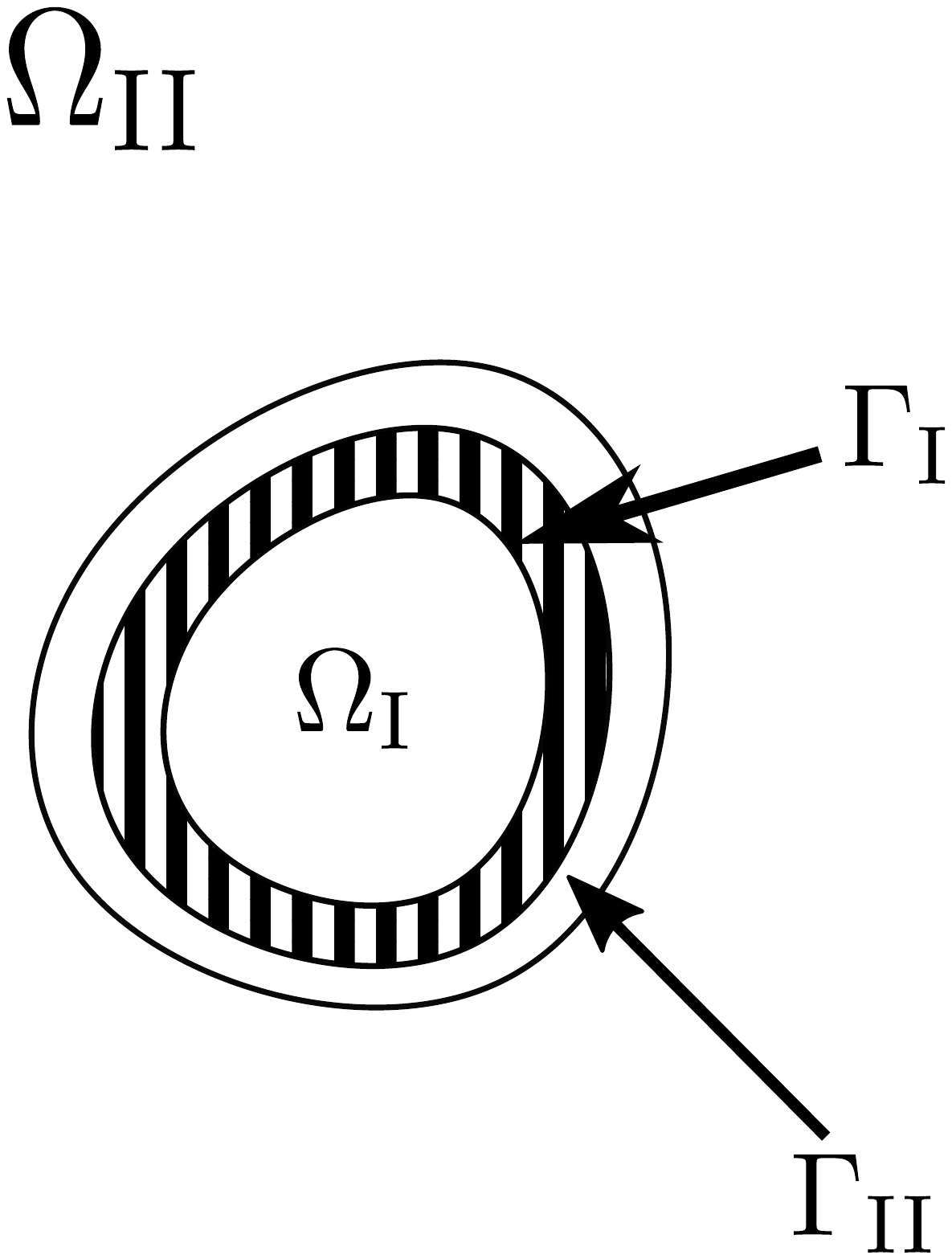}
  \caption{An illustration of the interface regions $\Gamma_\i$ and $\Gamma_\ii$.}
  \label{fig:2}
\end{figure}

The influence of $\bm u_\ii$ on $\bm u_\i$, $\bm f_\icii=-D_{\icii}\bm u_\ii$, corresponds to the Neumann boundary condition on $\Omega_\i$. Similarly, we define $\bm f_\iici=-D_{\iici}\bm u_\i$ as the influence of $\bm u_\i$ on $\bm u_\ii$. Due to the short-range interactions and the translational invariance of the force constant matrices among the atoms, $D_\icii$ is a sparse matrix, that is, most entries in the matrices $D_{\icii}$ and $D_{\iici}$ are zeros. The following notations are motivated by the domain decomposition method, and they are useful to reveal the sparsity of the matrices and the nearsightedness of the interactions. We define the sets of boundary atoms as $$\Gamma_\i = \{i \in \Omega_\i\vert D_{ij}\ne 0 , \text{ for any } j \in \Omega_\ii\}, \text{ and }\Gamma_\ii=\{i\in \Omega_\ii \vert D_{ij}\ne 0 , \text{ for any } j \in \Omega_\i\}.$$ They are the collections of the atoms at the inner and outer boundaries, as indicated in Fig. \eqref{fig:2}. We now use $\bm u_{\iciii}$ and $\bm u_{\iiciii}$ to represent the displacement of atoms in $\Gamma_\i$ and $\Gamma_\ii$, respectively. $\bm u_{\iciii}$ has a dimension of $m_\i$, and $\bm u_{\iiciii}$ has a dimension of $m_\ii$. We have $m_\i$, $m_\ii$ $\ll$ $n_\i$.  With these notations,  the short-range interactions can be revealed by writing the two matrices as 
\begin{equation}
  \label{eq:12}
  D_{\icii} =
  \begin{bmatrix}
    D_{\iciii} & \mathbf 0
  \end{bmatrix},
  \text{ and }
  D_{\iici} =
  \begin{bmatrix}
    D_{\iiciii} & \mathbf 0
  \end{bmatrix}.
\end{equation}
It is also natural to introduce the matrices $E_\i$ and $E_\ii$ such that $\bm u_{\iciii}=\bm u_\i\vert_{\Gamma_\i}=E_\i\bm u_\i$, and $\bm u_{\iiciii}=\bm u_\ii\vert_{\Gamma_\ii}=E_\ii\bm u_\ii$. Furthermore, we define $\bm f_{\iciii}=-E_\i D_{\iciii}\bm u_{\iiciii}$, and $\bm f_{\iiciii}=-E_\ii D_{\iiciii}\bm u_{\iciii}$ to denote the forces at the inner and outer boundaries. 

\medskip

\begin{table}[h]
  \centering
  {\footnotesize
    \begin{tabular}[h]{l|l|l}
      \hline
      equations & kernel functions & interpretations\\
      \hline
      $\bm U_{\iiciii} = B\bm U_{\iciii}$ & $B(s) = -E_\ii(s^2I+D_{\iicii})^{-1}D_{\iiciii}$ & \makecell[cl]{ABEM \cite{Li2012, Wu2017}\\ section: \ref{sec:eval}}\\
      \hline
      $\bm U_{\iiciii} = \Theta \bm F_{\iiciii}$ & $\Theta(s) = E_\ii(s^2I+D_{\iicii})^{-1}E_\ii^T$& \makecell[cl]{Target kernel \\ to approximate \\ NtD map}\\
      \hline
      $\bm F_{\iciii} = T \bm U_{\iciii},$  & $T(s) = D_{\iiciii}(s^2I+D_\iicii)^{-1}D_{\iiciii}^T$ & \makecell[cl]{$\text{DtN map in }\Omega_\ii$\\ \cite{Oberai1998}}\\
      \hline
      $\bm F_{\iciii} = K s\bm U_{\iciii} + \Theta(0)\bm U_{\iciii}$ & $\displaystyle K(s) = \frac{\Theta(s)-\Theta(0)}{s}$ & \makecell[cl]{GLE in $\Omega_\ii$\\ \cite{Doll1976, Adelman1974}}\\
      \hline
    \end{tabular}
    \caption{Kernel functions in the Laplace domain.}
    \label{tab:1}
  }
\end{table}  

We now turn to the general forms of the boundary conditions. There are various ways to express the boundary conditions based on  Eq. \eqref{eq:3}. Their forms in the Laplace domain are summarized in Table \ref{tab:1}, along with references where these representations can be found. We will briefly go through the derivation of these formulas.

The first kernel function can be immediately obtained if we left-multiply both sides of Eq. \eqref{eq:3} by $E_\ii$.  This kernel function corresponds to a mapping between the displacement at the inner and outer boundaries. 

\medskip
To derive the second representation, one notices that the influence of $\bm u_\i$ on $\bm u_\ii$ can be expressed as
\begin{equation}
  \label{eq:39}
  \bm f_{\iici} =
  \begin{bmatrix}
    \bm f_{\iiciii} \\ \bm 0
  \end{bmatrix}.
\end{equation}
Further notice that the right hand side of Eq. \eqref{eq:3} can be written as
\begin{equation}
  \label{eq:43}
  (s^2I+D_{\iicii})^{-1}\bm F_{\iici} =
  \begin{bmatrix}
    (s^2I+D_{\iicii})^{-1}E_\ii^T & * 
  \end{bmatrix}
  \begin{bmatrix}
    \bm F_{\iiciii}\\ \bm 0
  \end{bmatrix}
  = (s^2I+D_{\iicii})^{-1}E_\ii^T\bm F_{\iiciii},
\end{equation}
where $\bm F_{\iici}(s)=\mathcal L\{\bm f_{\iici}(t)\}$ is the Laplace transform of $\bm f_{\iici}(t)$, and $\bm F_{\icii}(s) =-D_{\icii}\bm u_{\ii}$ is the force from atoms in region $\Omega_\i$ to atoms in region $\Omega_\ii$. 

To proceed, we left-multiply both sides of Eq. \eqref{eq:3}, and we get
\begin{equation}
  \label{eq:44}
  \bm U_{\iiciii} = E_\ii(s^2I+D_{\iicii})^{-1}E_\ii^T\bm F_{\iiciii}.
\end{equation}
This way, only the atoms outside the boundary are involved in Eq. \eqref{eq:44}. This DtN map connects the displacement and traction at the outer boundary.  Numerically, the kernel function $\Theta(s)$ is amenable to numerical and analytical treatments, because it is symmetric positive definite for all $s$. (It is a principle submatrix of $(s^2I+D_\iicii)$.) It is much more convenient to approximate a symmetric positive definite matrix-valued matrix by interpolations. The matrix $(s^2I+D_\iicii)$ corresponds to a continuum screened Poisson operator or Klein-Gordon operator \cite{Alhaidari2006}.

\medskip

To arrive at the third expression, one can left-multiply $D_{\iiciii}$ to both sides of Eq. \eqref{eq:3}, which leads to,
\begin{equation}
  \label{eq:4}
  \bm F_{\iciii}(s)=D_{\iiciii}(s^2I+D_{\iicii})^{-1}D_{\iiciii}^T\bm U_{\iciii}(s).
\end{equation}
In this case, the DtN operator is given by $$T(s)=D_{\iiciii}(s^2I+D_{\iicii})^{-1}D_{\iiciii}^T,$$ and it is a mapping between the displacement and traction at the inner boundary.

In the real-time domain, this DtN map is written as
\begin{equation}
  \label{eq:5}
  \bm f_{\iciii}(t) = \int_0^t\tau(t-s)\bm u_{\iciii}(s)ds.
\end{equation}
Here $\tau(t) = \mathcal L^{-1}\{T(s)\}$ is the inverse Laplace transform of $T(s)$. Fourier transformation versions of the expression \eqref{eq:5} were discussed in \cite{Adelman1974, Doll1976, Li2006}.

\medskip
Finally, the fourth expression in Table 1  can be  obtained by writting,
\begin{equation}
  \label{eq:16}
   \bm F_{\iciii}(s) = \frac{T(s)-T(0)}{s}(s\bm U_{\iciii}(s)-\bm U_{\iciii}(0)) + T(0)\bm U_{\iciii}(s),
 \end{equation}
 and denote $K(s)=(T(s)-T(0))/s$. In the real-time domain, $\bm f_{\iciii}$, accordingly,  is expressed in terms of the {\it velocity} of the system. The alternative form of $\bm f_{\iciii}$ is 
 \begin{equation}
   \label{eq:17}
   \bm f_{\iciii} = \int_0^t\kappa(t-s)\dot{\bm u}_{\iciii}(s)ds - T(0)\bm u_{\iciii}(t).
 \end{equation}

 In the time domain, this derivation is equivalent to an integration by parts. In this form, the memory term can be viewed as a damping term, and it makes the stability analysis more straightforward \cite{Li08a}. 
 
 \medskip
 
In this paper, we choose the {\it second} expression of the DtN map, with kernel function $\Theta(s)$, to represent the boundary condition. 
The main reason is that the kernel function $\Theta(s)$ is symmetric positive definite for any value of $s> 0$, which may not hold in other three cases. Due to this property, some theoretical results can be easily established, and they are quite relevant to the stability of the boundary condition, which will be discussed in section \ref{sec:stability}. For convenience, we denote $\Theta_0 = \Theta(0)$, $\Theta_1=\Theta(s_1)$, and $\Theta_2=\Theta(s_2)$ for positive values of $s_1$ and $s_2$.

\begin{lemma}
  \label{lemma:1}
  For any  $s_1 \ne s_2 >0$, $(s_1^2(\Theta_1-\Theta_0)^{-1}-s_2^2(\Theta_2-\Theta_0)^{-1})/(s_1-s_2)$ is symmetric negative definite.
\end{lemma}
\begin{proof}
  Let us consider
  \begin{equation}
    \label{eq:37}
    M = \frac{(s_1^2I+D_{\iicii})^{-1}-(s_2^2I+D_{\iicii})^{-1}}{s_1-s_2}.
  \end{equation}
  which is symmetric.
  Now, we show $M$ is also negative definite. In fact, we can write $M$ into a product of two matrices,
   \begin{equation}
     \label{eq:27}
     \begin{aligned}      
       M = (s_1^2I+D_{\iicii})^{-1}\frac{s_2^2-s_1^2}{s_1-s_2}(s_2^2+D_{\iicii})^{-1}.
     \end{aligned}
   \end{equation}
   Both two matrices $(s_1^2I+D_{\iicii})^{-1}$ and $(s_2^2I+D_{\iicii})^{-1}$ are symmetric positive definite. Since $M$ is symmetric, it follows immediately that $M$ and $E_\ii ME_\ii^T$ are symmetric negative definite.
   Now, we consider
   \begin{equation}
     \label{eq:25}
     \begin{aligned}
       &\frac{s_1^2(\Theta_1-\Theta_0)^{-1}-s_2^2(\Theta_2-\Theta_0)^{-1}}{s_1-s_2}\\
      = & (\Theta_1-\Theta_0)^{-1}\frac{s_1^2(\Theta_2-\Theta_0)-s_2^2(\Theta_1-\Theta_0)}{s_1-s_2}(\Theta_2-\Theta_0)^{-1}\\
      = &(\Theta_1-\Theta_0)^{-1}s_1^2s_2^2E_\ii ME_\ii^T D_{\iicii}^{-1}(\Theta_2-\Theta_0)^{-1}.
    \end{aligned}
  \end{equation}
  Here $E_\ii ME_\ii^TD_{\iicii}^{-1}$ is negative definite. In addition, $(\Theta_1-\Theta_0)^{-1}$ and $(\Theta_2-\Theta_0)^{-1}$ are symmetric negative definite. As a result, $(s_1^2(\Theta_1-\Theta_0)^{-1}-s_2^2(\Theta_2-\Theta_0)^{-1})/(s_1-s_2)$ is symmetric negative definite.
\end{proof}

\begin{lemma}
  \label{lemma:2}
  If $0< s_1\ne s_2 < \sqrt{\lambda_{\text{min}}(D_{\iicii})}$, then $s_1(\Theta_2-\Theta_0)^{-1}-s_2(\Theta_1-\Theta_0)^{-1}/(s_1-s_2)$ is positive definite.
\end{lemma}
\begin{proof}
  We will use the same technique as the proof of the last lemma. Let us consider
  \begin{equation}
    \label{eq:38}
    N = \frac{s_1(s_1^2I+D_{\iicii})^{-1}-s_2(s_2^2I+D_{\iicii})^{-1}}{s_1-s_2},
  \end{equation}
  which is symmetric.
  
  Left-multiplying $N$ by $(s_1^2I+D_{\iicii})$ and right-multiplying $N$ by $(s_2I+D_{\iicii})$, we get
  \begin{equation}
    \label{eq:29}
    \begin{aligned}
      N=& (s_1^2I+D_{\iicii})^{-1}\frac{s_1(s_2^2I+D_{\iicii})-s_2(s_1^2I+D_{\iicii})}{s_1-s_2}(s_2^2+D_{\iicii})^{-1}\\
      =&  (s_1^2I+D_{\iicii})^{-1}\frac{(s_1-s_2)(-s_1s_2I+D_{\iicii})}{s_1-s_2}(s_2^2+D_{\iicii})^{-1}.\\
    \end{aligned}
  \end{equation}
  Since $s_1< \sqrt{\lambda_{\text{min}}(D_{\iicii})}$ and $s_2<\sqrt{\lambda_{\text{min}}(D_{\iicii})}$,
  $N$ is symmetric positive definite. Therefore, 
  \begin{equation}
    \label{eq:28}
    \begin{aligned}
      &\frac{s_1(\Theta_2-\Theta_0)^{-1}-s_2(\Theta_1-\Theta_0)^{-1}}{s_1-s_2}\\
      =& (\Theta_1-\Theta_0)^{-1}\frac{s_1(\Theta_2-\Theta_0)-s_2(\Theta_1-\Theta_0)}{s_1-s_2}(\Theta_2-\Theta_0)^{-1}\\
      =& (\Theta_1-\Theta_0)^{-1}s_1s_2E_\ii NE_\ii^T D_{\iicii}^{-1}(\Theta_2-\Theta_0)^{-1}\\
    \end{aligned}
  \end{equation}
  is symmetric positive definite.
\end{proof}
\section{Approximations of the DtN map $\Theta $ }
\label{sec:approx0}
The direct evaluation of the DtN operator $\Theta$ is quite challenging. The difficulties of computation lie in two aspects: the inverse of a large matrix and infinitely many evaluations for frequency parameter $s$. In this section, the exact DtN map will be approximated by rational functions (section \ref{sec:approx}). We avoid the inverse of the large matrix by a domain reduction of the system. The domain reduction, in the form  of the DtN map in the Laplace domain, will be discussed in section \ref{sec:eval}. 
\subsection{Approximation by Rational Functions}
\label{sec:approx}
Instead of evaluating the time-dependent memory kernel, rational approximations can be made \cite{Engquist1977}, which leads  to  dynamics without memory. In fact,  the memory is eliminated via the introduction of a new variable. 
We consider the general form of the rational functions,
\begin{equation}
  \label{eq:41}
  R_{n,n}(s)=(s^{n}I-s^{n-1}B_0-\cdots-B_{n-1})^{-1}(s^{n-1}A_0+\cdots+A_{n-1}), 
\end{equation}
where $A_i, B_i \in \mathbb R^{m_\ii\times m_\ii}$ for all $0\leq i<n$.  Notice that the dimension of the matrices $A_i$ and $B_i$ is small: $m_\ii\ll n_\ii$. The rational functions should satisfy $R_{n,n}(s)\approx\Theta(s)$. The form of the rational functions is inspired by the properties of the kernel function $\Theta(s)$: $\lim_{s\rightarrow \infty}\Theta(s) = 0$, and $\Theta'(0)=0$. In our approximation, $\lim_{s\rightarrow \infty}R_{n,n}(0)=0$ is automatically satisfied. For all following approximations, we always assume $R_{n,n}(0)=\Theta(0)$. This condition gives us
\begin{assumption}
  \label{assump:1}
  $-B_{n-1}^{-1}A_{n-1}=\Theta_0.$
\end{assumption}

For higher order approximations, we may assume $R_{n,n}'(0)=\Theta'(0)$ to obtain the better approximation around zero. The first derivative condition leads to
\begin{assumption}
  \label{assump:2}
  $-B_{n-1}^{-1}B_{n-2}B_{n-1}^{-1}A_{n-1}+B_{n-1}^{-1}A_{n-2}=0$.
\end{assumption}

In this paper, we only show zeroth, first, and second order approximations. It is straightforward to extend the idea to high order cases. We define the approximation order as the degree of the polynomial in the denominator. 

{\bf Zeroth order approximation.} The zeroth order approximation is to use a constant matrix, which, according to Assumption \ref{assump:1}, becomes $R_{0,0}=\Theta(0)$. One can also choose $s=s_0$ ($s_0$ is nonzero) as long as the corresponding dynamics is stable. In this paper, we make the intuitive choice $s=0$. The approximate dynamics is reduced to
\begin{equation}
  \label{eq:0th}
  \ddot{\bm u}_{\i}=-D_{\ici}\bm u_\i+E_\i^TD_{\iiciii}^TD_{\iicii}^{-1}D_{\iiciii}\bm u_{\iciii}.
\end{equation}
The stability of this approximation is shown in section \ref{sec:stability}. The dynamics can be further simplified by introducing $\overbar D_{\ici}=D_{\ici}-E_\i^TD_{\iiciii}^TD_{\iicii}^{-1}D_{\iiciii}E_\i$. In practice, the interactions in $\Omega_\i$ are nonlinear, and $\overbar D_\ici$ never has to be computed. This will be discussed in the next section.  Nonetheless, this Schur complement form gives us the insight of the stability of the resulting dynamical system. The zeroth order approximation only involves the displacement of atoms in region $\Omega_\i$. 

{\bf First order approximation.} In this case, we consider the rational function,
\begin{equation}
  \label{eq:R11}
  R_{1,1}(s)=(sI-B_0)^{-1}A_0,
\end{equation}
as the approximation. The coefficients $A_0$ and $B_0$ are determined by the interpolation at points $\left(s_0, \Theta(s_0)\right)$ and  $\left(0, \Theta(0)\right)$, in accordance with  Assumption \ref{assump:1}.
The corresponding approximate DtN map in real-time domain can be written as
\begin{equation}
  \label{eq:7}
  \begin{aligned}
    \bm u_{\iiciii} &= \int_0^te^{sB_0}A_0\bm f_{\iiciii}(t-s)ds\\
    &= -B_0^{-1}A_0\bm f_{\iiciii}(t) + \int_0^te^{sB_0}B_0^{-1}A_0\dot{\bm f}_{\iiciii}(t-s)ds,
  \end{aligned}
\end{equation}
which is similar to the alternative form (Eq. \eqref{eq:17}) of the exact DtN map. We reformulate the approximate DtN map by integration by parts for the purpose of stability analysis in section \ref{sec:stability}. It is worth pointing out that, in this setting, $\bm u_{\iiciii}$ is expressed in terms of $\dot{\bm f}_{\iiciii}$. The term $-B_0^{-1}A_0\bm f_{\iiciii}$, which according to Assumption \ref{assump:1}, coincides with $\Theta(0) \bm f_{\iiciii}$ that  will be moved to the first equation in Eq. \eqref{eq:lds}. We denote $\bm g=\int_0^te^{sB_0}B_0^{-1}A_0\dot{\bm f}_{\iiciii}(t-s)ds$ and append the dynamics of $\bm g$ to that of $\bm u_\i$. The variable $\bm g$ ($\bm g \in \mathbb R^{m_\ii\times 1}$) is coupled with $\bm v_{\iciii}$ only. With direct computation, one can verify that the approximate dynamics in this case is expressed as
\begin{equation}
  \label{eq:1st}
  \left\{
    \begin{aligned}
      \ddot{\bm u}_{\i} &=-\overbar D_{\ici}\bm u_\i - D_{\iciii}\bm g,\\
      \dot{\bm g}&= -B_0^{-1}A_0D_{\iiciii}\bm v_{\iciii} + B_0\bm g,
    \end{aligned}
  \right.
\end{equation}
where $\overbar D_\ici$ has the same definition as the one in the zeroth order approximation.

{\bf Second order approximation.}
We express the second order rational approximation as
\begin{equation}
  \label{eq:11}
  R_{2,2}(s)=(s^2I-sB_0-B_1)^{-1}(sA_0+A_1).
\end{equation}
The coefficients $A_0$, $A_1$, $B_0$, and $B_1$ are determined by the interpolation among the values of $\Theta(s)$ and $\Theta'(s)$. In the second order approximation, we take both assumption \ref{assump:1} and assumption \ref{assump:2} into consideration. Two more values of $\Theta(s)$ will be used, at points $(s_1, \Theta(s_1))$ and $(s_2,\Theta(s_2))$, to determined those coefficients. In the real time domain, the corresponding approximate dynamics is
\begin{equation}
  \label{eq:8}
  \ddot{\bm u}_{\iiciii} = B_0\dot{\bm u}_{\iiciii} + B_1\bm u_{\iiciii} + A_0\dot{\bm f}_{\iiciii} + A_1\bm f_{\iiciii}.
\end{equation}
Equivalently, the above dynamics can be expressed as
\begin{equation}
  \label{eq:9}
  \dot{\bm w}= B\bm w + A\bm f_{\iiciii},
\end{equation}
where $\bm w =
\begin{bmatrix}
  \bm u_{\iiciii} \\ \bm z
\end{bmatrix}
$,
$ B = 
\begin{bmatrix}
  B_0 & I\\ B_1 & \mathbf 0
\end{bmatrix},
$ and $ A = 
\begin{bmatrix}
  A_0\\ A_1
\end{bmatrix}
$. Employing the same technique in Eq. \eqref{eq:7}, we obtain
\begin{equation}
  \label{eq:10}
  \bm w = -B^{-1}A\bm f_{\iiciii}(t) + \int_0^te^{sB}B^{-1}A\dot{\bm f}_{\iiciii}(t-s)ds.
\end{equation}
We denote $\bm g=\int_0^te^{sB}B^{-1}AD_{\iiciii}\dot{\bm f}_{\iciii}(t-s)ds$ and introduce the matrix $E_1$ (different from $E_\i$) such that $\bm f_{\icii}=E_1\bm g$. Then we can write an extended dynamics to represent the ABC:
\begin{equation}
  \label{eq:2nd}
  \left\{
    \begin{aligned}
      \ddot{\bm u}_{\i} &=-\overbar D_{\ici}\bm u_{\i} - E_1 D_{\iciii}\bm g,\\
      \dot{\bm g}&= -B^{-1}AD_{\iiciii}\bm v_{\iciii} + B\bm g.
    \end{aligned}
  \right.
\end{equation}

\subsection{Evaluation of the DtN map}
\label{sec:eval}
In practice, the  sub-region $\Omega_{\ii}$ contains many atoms, but the region $\Omega_\i$ contains much fewer atoms. As a result, the matrix $\Theta(s) \in \mathbb R^{m_\ii\times m_\ii}$ associated with the DtN map does not have a large dimension. However, it involves  the inverse of $(s^2I+D_{\ii}) \in \mathbb R^{n_\ii\times n_\ii}$. The direction computation can be extremely expensive.

Fortunately, it is unnecessary to compute the entire inverse of $(s^2I+D_{\iicii})$. We only need to compute the Schur complement of the corresponding block of $(s^2I+D_{\iicii})$ for the  atoms at the interface $\Gamma$. Thanks to the translation invariance of the force constant matrices, this calculation can be done very efficiently. We first consider the problem in the Laplace domain,
\begin{equation}
  \label{eq:14}
  \widetilde D_{\iicii}(s)\bm U_\ii(s) = \bm F_{\iici}(s), \text{ for given } \bm F_{\iici}(s) = -D_{\iici}\bm U_\i(s),
\end{equation}
where $\widetilde D_{ij} = D_{ij} + s^2\delta_{ij}I$. With the help of the lattice Green's function (\ref{sec:lgf}), the equation can be reduced to the degree of freedoms on the interface $\Gamma$ by the atomistic-based boundary element method (ABEM) \cite{Li2012, Wu2017}. In the original work, ABEM is implemented for static elasticity. In this paper, we extend the idea to dynamics problems where the force constant  matrices are shifted by $s^2 I$. 

Since $\widetilde D = [\widetilde D_{ij}]$ is positive definite, the corresponding lattice Green's function is well-defined, which follows the relation \( \sum_j\widetilde G_{nj}\widetilde D_{ji}=\delta_{ni}I.\) The notation, $\widetilde \cdot$, represents the variable in the Laplace domain. As a result, the displacement of atom $n$ ($n\in \Omega_\ii$) can be trivially expressed as
\begin{equation}
  \label{eq:34}
  \bm U_n = \sum_{i\in\Omega_{\ii}}\delta_{ni}\bm U_i =  \sum_{i\in\Omega_{\ii},j}\widetilde G_{nj}\widetilde D_{ji}\bm U_i.
\end{equation}
The key step of the dimensional reduction is applying Abel's lemma (summation by parts) to Eq. \eqref{eq:34}. In this case, the Abel's lemma is expressed as
\begin{equation}
  \label{eq:35}
  \sum_{i\in\Omega_\ii,j}\widetilde G_{nj}\widetilde D_{ji}\bm U_i = \sum_{i\in \Gamma_\i,j\in \Gamma_\ii}\widetilde G_{ni}D_{ji}\bm U_j-\sum_{i\in \Gamma_\i,j\in \Gamma_\ii}\widetilde G_{nj}D_{ji}\bm U_i + \sum_{j\in \Omega_\ii}\widetilde G_{nj}\bm B_j,
\end{equation}
where $\bm B_j = \sum_{i\in \Omega} D_{ji}\bm U_i$. In the above equation, the first two summations are over the interface $\Gamma_\i$ and $\Gamma_\ii$ because of the locality of $\widetilde D_{ij}$. When no external force is present, $\bm B_j=0$. 
If we choose $n\in \Gamma_\ii$, Eq. \eqref{eq:35} forms a linear system,
\begin{equation}
  \label{eq:36}
  \bm U_{\iiciii}=\widetilde K\bm U_{\iiciii} + \widetilde L\bm F_{\iiciii},
\end{equation}
where
\begin{equation}
  \label{eq:46}
  \left\{
    \begin{aligned}
      \widetilde K_{nj} &= \sum_{i\in \Gamma_\i}\widetilde G_{ni}D_{ji}\\
      \widetilde L_{nj} &= \widetilde G_{nj}
    \end{aligned}
  \right.
\end{equation}
The linear system is still closed. This linear system provides an alternative expression of the DtN map, given by 
\begin{equation}
  \label{eq:15}
  \bm U_{\iiciii}(s) = (I-\widetilde K(s))^{-1}\widetilde L(s)\bm F_{\iiciii}(s).
\end{equation}
Notice that the matrices in the linear system have dimensions much smaller than $n_\ii.$ Eq. \eqref{eq:15}, which provides the Schur complement of the corresponding block of the matrix $(s^2I + D_\iicii)$, is equivalent to Eq. \eqref{eq:44}. In Eq. \eqref{eq:15}, we do not need to evaluate the inverse of a large matrix. 
\section{Stability of Absorbing Boundary Conditions (ABCs)}
\label{sec:stability}

As a Hamiltonian system, the stability of the MD model after modifications due to the approximate BCs is a very delicate issue \cite{Perko2013}. 
In this section, we will provide the principles of the approximations in section \ref{sec:approx} to ensure stability. 
Since the coefficients of the rational approximation are determined by interpolation, our principles will focus on the selections of interpolation points.

{\bf Zeroth order approximation.} The zeroth order approximation is automatically stable when we choose constant matrix $\Theta(0)$ as $R_{0,0}(s)$. In fact, $(D_{\i,\i}-E_\i^TD_{\iiciii}^TD^{-1}_{\iicii}D_{\iiciii})$ is the Schur complement of $D_{\iicii}$ of symmetric positive definite matrix $D$. Therefore, we have the following stability condition of the dynamics \eqref{eq:1st}.
\begin{theorem}[Zeroth order approximation]
The zeroth-order approximate dynamics (Eq. \eqref{eq:0th}) is stable, provided that the interpolation point is $(0, \Theta(0))$.
\end{theorem}

{\bf First order approximation.} To establish the stability of Eq. \eqref{eq:1st}, we introduce the following Lyapunov functional:
\begin{equation}
  \label{eq:6}
  E(\bm u_\i, \bm v_\i, \bm g) = \frac{1}{2}\bm v_\i^T\bm v_\i+\frac{1}{2}\bm u_\i^T\overbar D_{\ici}\bm u_\i - \frac{1}{2}\bm g^TA_0^{-1}B_0\bm g.
\end{equation}
Since $-A_0^{-1}B_0$ (Assumption \ref{assump:1}) is symmetric positive definite, $E\geq0$ for any $\bm v_\i$, $\bm u_\i$, and $\bm g$. The derivative of $V$, 
\begin{equation}
  \label{eq:18}
  \begin{aligned}
    L_tE &= \bm v_\i^T\dot{\bm v}_\i+\bm u_\i^T\overbar{D}_{\ici}\dot{\bm u}_\i-\bm g^TA_0^{-1}B_0\dot{\bm g}\\
    & = \bm v_\i^T(-\overbar{D}_\ii\bm u_\i-D_{\iciii}\bm g) + \bm u_i^T\overbar{D}_\ici\dot{\bm v}_\i-\bm g^TA_0^{-1}B_0(-B_0^{-1}A_0D_{\iiciii}\bm v_\i + B_0\bm g)\\ 
    &= -\bm g^TA_0^{-1}B_0^2\bm g.
  \end{aligned}
\end{equation}
Hence, when $A_0^{-1}B_0^2$ is a positive semi-definite matrix, and $B_0^{-1}A_0$ (Assumption \ref{assump:1}) is symmetric negative definite. According to standard ODE theory \cite{Perko2013}, the first order approximate dynamics \eqref{eq:1st} is stable. 

\begin{theorem}[First order approximation]
  \label{thm:1st}
  The approximate dynamics \eqref{eq:1st} is stable if the coefficients of the rational function (Eq. \eqref{eq:R11}) are determined by $R_{1,1}(0)=\Theta_0$ (Assumption \ref{assump:1}) and $R_{1,1}(s_1)=\Theta_1$ with any $s_1> 0$.
\end{theorem}
\begin{proof}
  The coefficients $A_0$ and $B_0$ are determined by solving the equations,
  \begin{equation}
    \label{eq:42}
    \left\{
      \begin{aligned}
        -B_0^{-1}A_0& = \Theta_0,\\
        (s_1I-B_0)^{-1}A_0 &= \Theta_1.
      \end{aligned}
    \right.
  \end{equation}
  We eliminate the coefficient $A_0$ and have
  \begin{equation}
    \label{eq:21}
    \left(\Theta_1-\Theta_0\right)B_0 = s_1\Theta_1.
  \end{equation}
  The matrix
  \begin{equation}
    \label{eq:22}
    \begin{aligned}
      \Theta_1-\Theta_0 &=E_\ii(s_1^2I+D_{\iicii})^{-1}E_\ii^T-E_\ii D_{\iicii}^{-1}E_\ii^T \\
      &= -s_1^2E_\ii(s^2_1+D_{\iicii})^{-1}D_{\iicii}^{-1}E_\ii^T,
    \end{aligned}
  \end{equation}
  is negative definite for any $s_1$. As a result, $B_0 = s_1(\Theta_1-\Theta_0)^{-1}\Theta_1$ is a negative definite matrix, and $A_0^{-1}B_0^2$ is a positive definite matrix. Therefore, the Lyapunov function defined by Eq. \eqref{eq:6} is nonnegative, and its derivative $L_tE$ is semi negative definite. 
\end{proof}

{\bf Second order approximation.} The stability of Eq. \eqref{eq:2nd} will be analyzed in a similar approach as described in the first order case. The Lyapunov functional for the system (Eq. \eqref{eq:2nd}) is defined by
\begin{equation}
  \label{eq:19}
  E(\bm u_\i, \bm v_\i, \bm g) = \frac{1}{2}\bm v_\i^T\bm v_\i + \frac{1}{2}\bm u_\i^T\overbar{D}_{\ici}\bm u_\i+\frac{1}{2}\bm g^TQ\bm g, 
\end{equation}
where
\begin{equation*}
  Q=
  \begin{bmatrix}
    -A_1^{-1}B_1 & \mathbf 0\\
    \mathbf 0 & A_1^{-1}
  \end{bmatrix}
  .
\end{equation*}
If $-A_1^{-1}B_1$ and $A_1^{-1}$ are symmetric positive definite matrices, then $Q$ is positive definite and the Lyapunov function is positive definite. The derivative of $E$ is
\begin{equation}
  \label{eq:20}
  \begin{aligned}
    L_tE &= \bm v_\i^T\dot{\bm v}_\i+\bm u_\i^T\overbar{D}_{\ici}\dot{\bm u}_\i+\bm g^TQ\dot{\bm g}\\
    & = -\bm v_\i^{T}E_1D_{\iciii}\bm g - \bm g^TQB^{-1}AD_{\iiciii}\bm v_\i + \bm g^TQB\bm g\\
    & = \bm g^TQB\bm g.
  \end{aligned}
\end{equation}
Here, $QB$ is explicitly expressed as
\begin{equation}
  QB = 
  \begin{bmatrix}
    -A_1^{-1}B_1B_0 & -A_1^{-1}B_1\\
    A_1^{-1}B_1 & \mathbf 0
  \end{bmatrix}.
\end{equation}
If $QB$ is negative semi-definite, then $L_tE < 0$. $QB$ is negative definite if and only if $-A_1^{-1}B_1B_0$ is negative definite. The stability conditions of dynamics \eqref{eq:2nd} are that $A_1^{-1}$ is symmetric positive definite, and $A_1^{-1}B_1B_0$ positive semi-definite.

Actually, these stability conditions are satisfied when we choose the interpolation points properly. 

\begin{theorem}[Second order approximation]
  The approximate dynamics \eqref{eq:2nd} is stable if the coefficients of the rational function Eq. \eqref{eq:11} are determined by $R_{2,2}(0)=\Theta(0)$, $R_{2,2}'(0)=\Theta'(0)$, $R_{2,2}(s_0)=\Theta(s_0)$, and $R_{2,2}(s_1)=\Theta(s_1)$ with $0< s_1 \neq s_2 < \sqrt{\lambda_{\text{min}}(D_{\iicii})}$. 
\end{theorem}
\begin{proof}
 In fact, the coefficients $A_0$, $A_1$, $B_0$, and $B_1$ are determined by solving a linear system,
 \begin{subequations}
   \begin{align}
     -B_1^{-1}A_1&=\Theta_0, \label{eq:2nd_a}\\
     B_1^{-1}B_0B_1^{-1}A_1-B_1^{-1}A_0&=0, \label{eq:2nd_b}\\
     s_1B_0\Theta_1+B_1\Theta_1+s_1A_0+A_1&=s_1^2\Theta_1, \label{eq:2nd_c}\\
     s_2B_0\Theta_2+B_1\Theta_2+s_2A_0+A_1&=s_2^2\Theta_2, \label{eq:2nd_d}                                  
   \end{align}
 \end{subequations}
 for the coefficients $A_0$, $A_1$, $B_0$, and $B_1$. Eq. \eqref{eq:2nd_b} (Assumption \ref{assump:2}) is reduced to $A_0=B_0B_1^{-1}A_1$. By solving the linear system for $B_0$ and $B_1$, we have
\begin{equation}
  \label{eq:24}
  B_0 = \frac{1}{s_1-s_2}\left(s_1^2\Theta_1(\Theta_1-\Theta_0)^{-1}-s_2^2\Theta_2(\Theta_2-\Theta_0)^{-1}\right),
\end{equation}
and
\begin{equation}
  \label{eq:23}
  \begin{aligned}
    B_1 &= \frac{s_1s_2}{s_2-s_1}\left(s_1\Theta_1(\Theta_1-\Theta_0)^{-1}-s_2\Theta_2(\Theta_2-\Theta_0)^{-1}\right).\\
  \end{aligned}
\end{equation}
By Lemma \ref{lemma:1}, $B_0$ is positive definite. Since $\Theta_0$ is symmetric positive definite, $A_1^{-1}B_1B_0$ is positive definite. By Lemma \ref{lemma:2}, $B_1$ is negative definite. Since $A_1=-B_1\Theta_0$, we have that the coefficient,
\begin{equation}
  \label{eq:26}
  A_1= s_1s_2\Theta_0 - \frac{s_1s_2}{s_2-s_1}\Theta_0\left[s_1(\Theta_1-\Theta_0)^{-1}-s_2(\Theta_2-\Theta_0)^{-1})\right]\Theta_0
\end{equation}
is symmetric positive definite. Therefore, the dynamics \eqref{eq:2nd} is stable.
\end{proof}
\section{Implementation of ABCs and partial-harmonic approximation of MD}
\label{sec:harmonic}

For practical applications,  the ABCs should be formulated so that the {\it nonlinear} interactions in $\Omega_\i$ are retained, to properly model defect structure, formation, and migration. 
This amounts to a {\it partial} harmonic approximation of the potential energy $V(\bm u_\i, \bm u_{\ii})$. In this approximation, only the interactions involving the atoms in $\Omega_\ii$ are linearized. From the exact potential, these linear interactions should have coefficients given by
\begin{equation}
  \label{eq:30}
  D_{\icii} = \frac{\partial^2V}{\partial \bm u_\i\partial \bm u_{\ii}}(\bm 0, \bm 0), \text{ and } D_{\iicii} = \frac{\partial^2V}{\partial \bm u_{\ii}^2}(\bm 0, \bm 0). 
\end{equation}
This is consistent with the notations in section \ref{sec:dtn}. In particular, $D_{\icii}$ represents in the coupling between $\Omega_\i$ and $\Omega_\ii.$

\medskip

The key observation is that when the region $\Omega_{\ii}$ is at mechanical equilibrium, we have $\bm u_{\ii}=C\bm u_{\i}$ with $C = -D_{\iicii}^{-1}D_{\iici}$, within the linear approximation. Notice that the equilibrium is relative to the displacement of the atoms in $\Omega_\i.$  With this mechanical equilibrium as the references, we introduce the partial harmonic approximation by defining the following approximate potential energy:
\begin{equation}
  \label{eq:31}
  \widetilde V(\bm u_\i, \bm u_{\ii}) = V(\bm u_\i, C\bm u_{\i}) + \frac{1}{2}(\bm u_{\ii}-C\bm u_{\i})^TD_{\iicii}(\bm u_{\ii}-C\bm u_\i).
\end{equation}
The first part is the potential energy in $\Omega_\i$ assuming a mechanical equilibrium in the surrounding area, while the second part is for region $\Omega_{\ii}$. 

We now show that this potential energy $\widetilde V$  is consistent with the exact model $V$ in the following sense.
\begin{enumerate}
\item {\bf $\widetilde V$ is a second order approximation.} More specifically, one can easily verify that,
 \begin{equation}
    \label{eq:49}
    \frac{\partial \widetilde V}{\partial u_i}(\bm 0, \bm 0) = 0, \text{ for any } i \in \Omega,
\end{equation}
and
\begin{equation}
  \label{eq:45}
  \frac{\partial^2\widetilde V}{\partial \bm u_i\partial \bm u_j}(\bm 0, \bm 0) = D_{ij} = \frac{\partial^2V}{\partial \bm u_i\partial \bm u_j}(\bm 0, \bm 0), \text{ for any } i, j \in \Omega.
\end{equation}
This of course implies that  $\widetilde V = V$ if $V$ is quadratic.

\item {\bf The interactions in the interior of $\Omega_\i$ is exactly preserved.} To see this property, we notice that most potential energy can be decomposed into site energies,
\begin{equation}
  \label{eq:47}
  V = \sum_{i\in \Omega}V_i,
\end{equation}
and $V_i$ only depends on the atoms that are within the cut-off radius around the $i$-th atom. One can easily verify that,
$V_i=\widetilde V_i$ for all $i\in \Omega_\i\setminus \Gamma_\i$.
\end{enumerate}

\medskip

Applying the Hamilton's principle, we obtain the corresponding approximate dynamics,
\begin{equation}
  \label{eq:32}
  \left\{
    \begin{aligned}
      \ddot{\bm u}_\i&=-\nabla_{\bm u_\i}V(\bm u_\i, C\bm u_\i) - C^T\nabla_{\bm u_{\ii}}V(\bm u_\i, C\bm u_\i) + C^TD_\iicii(\bm u_{\ii} - C\bm u_\i),\\
      \ddot{\bm u}_{\ii}&=-D_{\iicii}\bm u_{\ii}-D_{\iici}\bm u_\i.
    \end{aligned}
  \right.
\end{equation}

Now, one can make the observation that the dynamics in $\Omega_\ii$ is identical to that in \eqref{eq:lds}. Further more, the first equation is only coupled to $\bm u_\ii$ linearly. These observations show  that the DtN map can be formulated in the same way as in section \ref{sec:dtn}. 
{As an example, let us write out the molecular dynamics model in $\Omega_\i$, supplemented with the first order ABC,
  \begin{equation}
    \label{eq:48}
    \left\{
      \begin{aligned}
        \ddot{\bm u}_\i&=-\nabla \Phi(\bm u_\i)- D_{\iciii}\bm g,\\
        \dot{\bm g} & =-B_0^{-1}A_0D_{\iiciii}\bm v_{\iciii} + B_0\bm g,
      \end{aligned}
    \right.
  \end{equation}
  where $\Phi(\bm u_\i)\overset{\text{def}}{=}V_\i(\bm u_\i, CE_\i^T\bm u_{\iciii}) $ is the effective potential energy in $\Omega_\i$.

  It is also straightforward to establish the stability of the dynamics \eqref{eq:48}. We introduce this Lyapunov functional,
  \begin{equation}
    \label{eq:50}
    E = \Phi(\bm u_\i)+ \frac{1}{2}\bm v_\i^T\bm v_\i - \frac{1}{2}\bm g^TA_0^{-1}B_0\bm g.
  \end{equation}
  Its derivative is simply given by,
  \begin{equation*}
    \label{eq:51}
      L_tE = - \bm g^TA_0^{-1}B_0\bm g.
  \end{equation*}
  So the stability result still follows. 
}


Since $\Omega_\ii$ involves many atoms, it is impracticable to evaluate the second term on the right hand side of the first equation in \eqref{eq:32}. 
But this term can be neglected. This can be justified by using an expansion,
\begin{equation*}
  \nabla_{\bm u_{\ii}}V(\bm u_\i, C\bm u_\i) \approx D_\iicii C\bm u_\i + D_\iici\bm u_\i = \bm 0.
\end{equation*}


The modeling error depends on the deformation of the system in $\Omega_\ii$. In this paper, we will not provide the detailed proof. In fact,
\begin{equation*}
  \nabla_{\bm u_{\ii}}V(\bm u_\i, C\bm u_\i) = D_\iicii C\bm u_\i + D_\iici\bm u_\i  + H.O.T.
  \end{equation*}
  In light of the definition of the matrix $C$, the leading terms become zero. 

\section{Numerical simulations}
\label{sec:num}
In this section, we present the results from two molecular simulations with ABCs applied. The standard Verlet method is implemented for the time integration. The ABCs are simply discretized by the forward Euler method with the same time step size as the Verlet method. The time step of the simulations is set to be $76.2$ femtoseconds ($10^{-15} s$).  

\noindent{\bf System setup.} We choose the nonlinear region $\Omega_\i=[-50\AA, 50\AA]\times[-50\AA, 50\AA]$ and the surrounding region to be infinite $\Omega_\ii=\mathbb R^2\!\setminus\!\Omega_\i$. The region $\Omega_\i$ is filled with $3,265$ bcc iron atoms. The presence of the four corners does not impose any computational difficulty in our formulation. The periodic boundary condition is imposed in the z-direction to mimic a plane strain condition. Among the iron atoms, the interactions are modeled by the EAM potential \cite{Shastry1995}. As preparations, the force constant matrices are approximated by standard finite difference formulas. Due to the locality of the computed force constant matrices,  only $580$ atoms in $\Omega_\ii$ and $520$ atoms in $\Omega_\i$ are involved in the evaluation of the DtN map $T(s)$.  The calculation of the lattice Green's function is discussed in the Appendix. For the initial configuration, the initial velocity is given by
\begin{equation}
  \label{eq:33}
  v_x = 20\sin(x/3)e^{-0.1r}, \text{ and } v_y=20\sin(y/3)e^{-0.1r},
\end{equation}
where $r=\sqrt{x^2+y^2}$. The displacement of all atoms is zero initially.

\noindent{\bf Interpolation points}. In all simulations, we apply zeroth, first, and second order approximate ABCs and monitor the wave reflections. For the zeroth order approximation, the only interpolation point is chosen as $s_0=0.1312$ PetaHz. For the first and second approximations, we need to solve a linear system for the coefficients of rational function. In practice, Assumption \ref{assump:1} can be extended to $-B_{n-1}^{-1}A_{n-1}=\Theta(s_0)$, where $s_0$ does not have to be zero. The proofs of stability still hold. 

For the first order approximation, we choose $s_0=0.1312$ PetaHz and $s_1=0.2624$ PetaHz. For the second order approximation, $s_0=0.01312$ PetaHz, $s_1=0.1312$ PetaHz, and $s_2=0.3936$ PetaHz.

\noindent{\bf Performance}. In implementing the proposed ABCs, only matrix-vector multiplications are involved in each time step, since the coefficients are determined a priori. Computation times were recorded (Table \ref{tab:2}) when the three approximate ABCs are imposed on the same problem. We see that the higher order ABCs do not add up the computational cost significantly. 
\begin{table}[!htbp]
  \centering
  \begin{tabular}[!htbp]{c|c|c}
    \hline
    zeroth order & first order & second order\\
    \hline
    1979.108 s & 2201.072 s & 2783.020 s\\
    \hline
  \end{tabular}
  \caption{Computational time of the three approximate ABCs imposed on the same problem, same computer, for $200,000$ time steps.}
  \label{tab:2}
\end{table}

\subsection{Waves in homogeneous systems}
\begin{figure}[!htbp]
  \centering
  \includegraphics[scale=0.5]{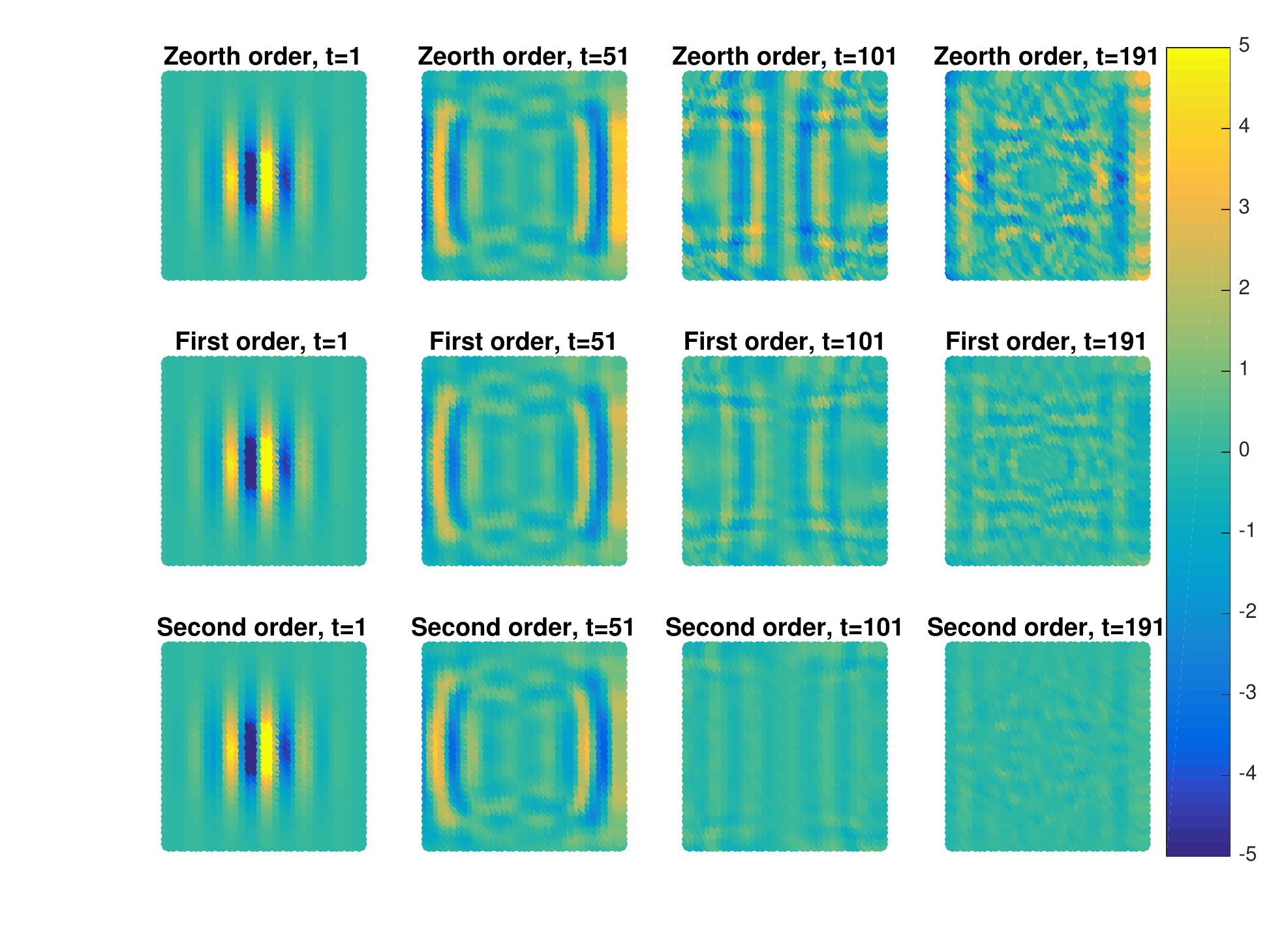}
  \caption{Wave propagations in the homogeneous system. The color indicates the velocity in the x direction.}
  \label{fig:waveperfect}
\end{figure}

We observe within a short time period, the resulting velocity fields, under the three ABCs,  are almost the identical. 
However, when  the lattice waves arrive at  the boundary, a complete reflection is observed under the zeroth order approximation. The reflection is reduced by the first order approximation. The second order approximation exhibited clear improvement: After several reflections, the lattice waves have been almost all absorbed. 

\subsection{Waves in a system with dislocations}
In the second experiment, we implement the ABCs in a system with dislocations. The dislocations are created at $(-20\AA,0)$ and $(20\AA, 0)$ by analytical solutions \cite{Liebowitz1968}. The Burgers vectors of the two dislocations are opposite in the same slip plane. In this case, the far-field displacement has fast decays. As preparations, the system is driven to an equilibrium state by minimizing the total energy. Then the same initial velocity is introduced. An interesting observation in this case is that  high-frequency waves are generated by the interaction of the initial lattice waves and  the dislocations. 
 We observed that the second order approximation still performs the best. The reduction of the reflections is still significant. 
 \begin{figure}[!htbp]
   \centering
   \includegraphics[scale=0.5]{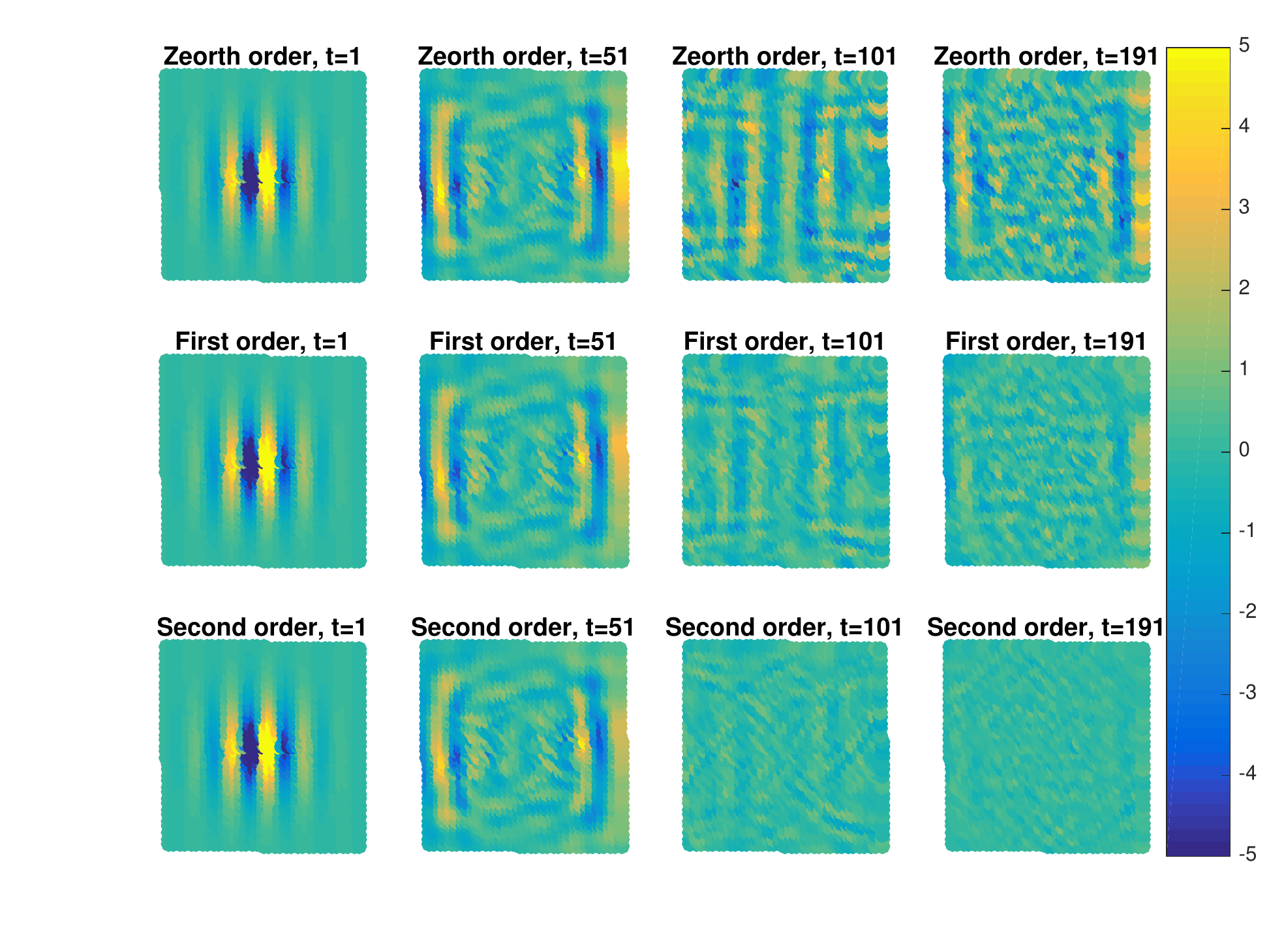}
   \caption{Wave propagations in the system with dislocations, the color indicates the value of velocity.}
   \label{fig:dislocwave}
 \end{figure}
 \section{Conclusion}
 In conclusion, we presented a new strategy to represent and approximate the ABCs systematically. It is a robust, efficient and easy-to-implement method to simulate phonons propagation in a large or infinite domain. Under this framework, further extensions can be pursued in various directions.
 
 {\it Finite temperature}. In this paper, we only focused on the wave propagation at low temperature. Following the original idea of Adelman and Doll's work is based on generalized Langevin equation (GLE), and the great deal of recent development, the proposed ABCs can be extended naturally to the case of finite temperature. The extended approach will be able to exchange heat between truncated atomistic region and the bath.
 
 {\it Dynamical loading}. Another important application of ABCs is to allow external elastic waves to come through the atomistic regions and interact with local defects, which is similar to \cite{Li08b}. Within the current framework, it means that a physical boundary needs to be included into the set $\Gamma$, where a   dynamic loading condition is applied. 
Another approach to  include dynamic loading and simulation the propagation of elastic waves is to couple molecular dynamics models to a continuum elastic wave equation \cite{LiYaE09,Wagner2003}. In such a setting, the transparent boundary condition is still an important component. 
 
 \appendix
 \label{sec:appendix}
 \section{Lattice Green's function}
 \label{sec:lgf}
 The calculations of static lattice Green's functions were discussed in previous works \cite{Li2012, Wu2017}. The integral expression of the lattice Green's function is numerically approximated by a quadrature formula over k-points. However, for the Green's function at long-distance, a direct calculation has to rely on a fine quadrature due to the high oscillations. Fortunately, it has been shown that the far-field lattice Green's function can be approximated by the continuum Green's function. The detailed proofs were provided in \cite{Martinsson2002, Ghazisaeidi2009}. In this paper, we focus on the calculation of the dynamics Green's function. An alternate form of the elastodynamic Green's function in Fourier space is given by C. -Y. Wang \cite{Wang1994}.
 
 To explain the connection, let us consider the following 2D Fourier transform in terms of wave number $\bm \xi$:
 \begin{equation}
   \begin{aligned}
     &F(\bm \xi)=\mathcal{F}[f(\bm x)]=\int_{\mathbb{R}^2}f(x)e^{-i\bm \xi\cdot \bm x}d\bm x,\\
     &f(\bm x)=\mathcal{F}^{-1}[F(\bm \xi)]=\frac{1}{4\pi^2}\int_{\mathbb{R}^2}F(\bm \xi)e^{i\bm \xi \cdot \bm x}d\bm \xi.
   \end{aligned}
 \end{equation}
 
 We are seeking the fundamental solution to the {\it discrete} problem,
 \begin{equation}
   \ddot{\bm u}_i+\sum D_{ik}\bm u_k = \bm f_i.
 \end{equation} The corresponding {\it continuum} problem is
 \begin{equation}
   (c_{ijkl}\partial_j\partial_l-\rho\delta_{jl}\partial_t^2)u_k = f_i,
 \end{equation}
 where $c_{ijkl}$ is the elastic constant. The elastic constant is connected to the force constant matrix, and the density $\rho$ is related to the volume of unit cell. We may take the Laplace transforms of two problems,
 \begin{equation}
   (\Gamma_{ik}(\bm \xi)+\rho s^2\delta_{ik})U_k= F_i,
 \end{equation}
 and 
 \begin{equation}
   (D(\bm \xi)+s^2I)\bm U_k = \bm F_i.
 \end{equation}
 Here $\Gamma_{ik}(\bm \xi) = c_{ijkl}\xi_j\xi_l$. $D(\bm \xi)=\sum D_R e^{-i\bm \xi\cdot \bm R}$ is the dynamic matrix for the discrete problem. Let us consider the Taylor expansion of $D(\bm \xi)$,  
\begin{equation}
  \begin{aligned}
    D(\bm \xi) &\approx  \sum_{\bm R}D_{\bm R}(1-\frac{1}{2}(\bm \xi\cdot \bm R)^2)\\
    &= (-\frac{1}{2}\sum_{\bm R} D_{i-k}R_jR_l\xi_j\xi_l)_{ik}.
  \end{aligned}
\end{equation}
The $\xi^0$ and $\xi^1$ terms disappear because of the inversion symmetry and the translation invariance. We see from direct comparison that the two problems are consistent with each other if 
\begin{equation}
  c_{ijkl}=-\frac{1}{2v_0}\sum_R D_{i-k}R_jR_l,
\end{equation}
where $v_0=1/\rho$ is the atomic unit volume.

Now, we show the connection between Green's functions of two above problems. The continuum Green's function is expressed as a Fourier integral,
\begin{equation}
  \label{eq:egf}
  G(\bm x,s)= \frac{1}{4\pi^2}\int_{\mathbb{R}^2} \big[\rho s^2 I + \Gamma(\bm\xi)]^{-1} \cos (\bm \xi \cdot \bm x) d\bm \xi,
\end{equation}
and the lattice Green's function is written as an integral over the first Brillouin zone,
\begin{equation}
  \label{eq:lgf}
  G(\bm R_j, s) = \frac{1}{|B|}\int_{B}[s^2I+D(\bm \xi)]^{-1}\cos(\bm \xi\cdot \bm R_j)d\bm \xi.
\end{equation}
Unlike static Green's functions, the matrices $s^2I+D(\bm \xi)$ and $\rho s^2 I + \Gamma(\bm\xi)$ are positive definite. Hence, both Green's functions are uniquely defined. Formally, the lattice Green's function \eqref{eq:lgf} converges to the continuum Green's function \eqref{eq:egf} when $\bm x$ and $\bm R$ are large enough, since only small $\bm \xi$ will contribute to the integral \eqref{eq:lgf} in this regime.

The continuum dynamic Green's function can be further simplified for numerical evaluation. When we turn it to the polar coordinate, the continuum Green's function can be written as
\begin{equation}
  G(\bm x,s)= \frac{1}{4\pi^2}\int_{0}^{2\pi} \int_0^{+\infty}
  \big[\rho s^2 I + r^2 \Gamma(\bm n)]^{-1} \cos (r \bm n \cdot \bm x) r dr d\theta, 
\end{equation}
where $\bm n=(\cos \theta, \sin \theta)$. With the eigenvalue decomposition, the above expression of the continuum Green's function is written as
\begin{equation}
  G(\bm x,s)= \frac{1}{4\pi^2}\int_{0}^{2\pi} 
  \sum_{i}\int_0^{+\infty}
  \frac{ \cos (r \bm n \cdot \bm x) r}{ \rho s^2  + r^2 \rho c_i^2(\bm n)}  dr \;\bm w_i(\bm n)\otimes \bm w_i(\bm n) d\theta.
\end{equation}
Here we wrote the $i$th eigenvalue of the matrix $\Gamma$ as $\rho c_i^2.$ The inner integral can be simplified to
\begin{equation}
  \frac{1}{\rho c_i^2} \int_0^{+\infty} \frac{\cos (t |\bm n \cdot \bm x| s/c_i)}{1+t^2}dt 
  = \frac{1}{\rho c_i^2} \Big[\sinh(a_i) \text{Shi}(a_i)  - \cosh(a_i) \text{Chi}(a_i)
  \Big]
\end{equation}
with $a=|\bm n \cdot \bm x| s/c_i.$ In this formula, both $\sinh(a_i)\text{Shi}(a_i)$ and $\cosh(a_i)\text{Chi}(a_i)$ are increasing exponentially when $a_i>0$. This expression may give us great round-off errors when we perform numerical simulations. But we can rearrange and simplify this expression as follows: By the properties of exponential integrals, $\text{Shi}(x) - \text{Chi}(x) = -\text{Ei}(-x)$, and $\text{Shi}(x) + \text{Chi}(x) = \text{Ei}(x)$ when $s$ is positive, we have another form of the inner integral,
\begin{equation}
  -\frac{1}{2\rho c_i^2} \Big[e^{a_i} \text{Ei}(-a_i)  + e^{-a_i} \text{Ei}(a_i)\Big],
\end{equation} 
where $e^{a_i} \text{Ei}(-a_i)$ and  $e^{-a_i} \text{Ei}(a_i)$ go to zero when $\|\bm x\|$ goes to infinity. This is more amenable to numerical evaluations. 

It is interesting that this continuum Green's function is also connected to Wang's formula \cite{Wang1994}, where the dynamic Green's function in the time domain is given by
\begin{equation}
  \label{eq:wang}
  g(\bm x, t) = \frac{H(t)}{8\pi^2}\int_{|\bm n|=1}\sum_i\frac{P^i(\bm n)}{\rho c_i}\left(\frac{1}{c_it+|\bm n\cdot \bm x|}+\frac{1}{c_it-|\bm n\cdot \bm x|}\right)d\bm n.
\end{equation}
Here $P^i(\bm n) = \bm w_i(\bm n)\otimes \bm w_i(\bm n)$. Even though our derivation differs from Wang's, both formulas are the integrals over unit circle. In fact, we take the Laplace transform of Eq. \ref{eq:wang} when $|\bm n\cdot \bm x| \neq 0$,
\begin{equation}
  G(\bm x, s) = -\frac{1}{4\pi^2}\int_0^{2\pi}\sum_i\frac{P^i}{\rho c_i^2}\Big[\frac{e^{a_i}\text{Ei}(-a_i) + e^{-a_i}\text{Ei}(a_i)}{2}\Big]d \theta.
\end{equation}
Note that there is a singularity when $\bm n\cdot \bm x = 0$.

In the implementation, the singularity can be removed by adding another term inside the integral. More specifically,
\begin{equation}\label{eq: A16}
  G(\bm x, s) = -\frac{1}{4\pi^2}\int_0^{2\pi}\sum_i\frac{P^i}{\rho c_i^2}\Big[\frac{e^{a_i}\text{Ei}(-a_i) + e^{-a_i}\text{Ei}(a_i)}{2} - \ln(|\bm n\cdot \bm x|)\Big]d \theta - G(\bm x, 0),
\end{equation}
where
\begin{equation}
  \begin{aligned}
    G(\bm x,0) &= -\frac{1}{4\pi^2}\int_0^{2\pi}\sum_i\frac{P^i}{\rho c_i^2}\ln(|\bm n\cdot \bm x|)d \theta \\
    & =-\frac{1}{4\pi^2\rho}\int_0^{2\pi}{\Gamma^{-1}(\bm n)}\ln(|\bm n\cdot \bm x|)d \theta, 
  \end{aligned}
\end{equation}
is the static Green's function for anisotropic solids, which according to 
the well known Stroh's formalism, can be expressed explicitly as,
\begin{equation}
  G(\bm x,0) = \frac{1}{\pi}\text{Im}(\ln(z_1)\bm a_1\otimes \bm a_1+\ln(z_2)\bm a_2\otimes \bm a_2)+C,
\end{equation}
where $z_1 = x_1+\mu_1x_2$ and $z_2 = x_1+\mu_2x_2$. $\mu_1$ and $\mu_2$ are the two roots to the sextic equation of elasticity, whose imaginary parts are positive. $\bm a_1$ and $\bm a_2$ are the two corresponding eigenvalues, and $C$ is a constant. Using this subtraction, the integral  in \eqref{eq: A16} can be easily evaluated. 
\section*{References}
\bibliographystyle{plain}
\bibliography{abcs_bbt}
\end{document}